\documentclass[letterpaper]{article} 
\usepackage[hyphens]{url}  
\usepackage{graphicx} 
\usepackage{natbib}  
\usepackage{caption} 
\usepackage{algorithm}
\usepackage{algorithmic}

\usepackage{newfloat}
\usepackage{listings}
\DeclareCaptionStyle{ruled}{labelfont=normalfont,labelsep=colon,strut=off} 
\floatstyle{ruled}
\newfloat{listing}{tb}{lst}{}
\floatname{listing}{Listing}

\usepackage{booktabs}

\usepackage{amsmath,amssymb,amsthm}
\usepackage{mathtools}
\usepackage{geometry}
\usepackage{hyperref}
\usepackage{booktabs}
\usepackage{enumitem}
\usepackage{xcolor}
\usepackage{microtype}
\usepackage{bm}
\usepackage{array}

\newtheorem{theorem}{Theorem}[section]

\newtheorem{corollary}[theorem]{Corollary}

\newtheorem{remark}[theorem]{Remark}
\newtheorem{definition}[theorem]{Definition}

\newcommand{\E}{\mathbb{E}}

\newcommand{\R}{\mathbb{R}}
\newcommand{\Cov}{\operatorname{Cov}}

\newcommand{\tr}{\operatorname{tr}}

\newcommand{\norm}[1]{\left\lVert #1 \right\rVert}

\newcommand{\Regret}{\mathcal{R}}
\newcommand{\phat}{\hat{\pi}}
\newcommand{\pstar}{\pi^{*}}
\newcommand{\calZ}{\mathcal{Z}}
\newcommand{\calC}{\mathcal{C}}
\newcommand{\calP}{\mathcal{P}}
\newcommand{\calL}{\mathcal{L}}   

\newcommand{\CovFun}{\mathbf{C}}

\newcommand{\grad}{\nabla}
\newcommand{\Hess}{\mathbf{H}}

\newcommand{\Sig}{\Sigma}

\title{Optimizing Regret}
\author {
    Irene Aldridge \footnote{
    Risk AI Center\\
    irene@riskaicenter.com
}
}

\date{}

\begin{document}

\maketitle

\begin{abstract}
Building on the identity that expected regret equals the covariance between costs and decisions, this paper develops a derivative theory of the covariance regret functional. We derive the G\^ateaux derivative, showing that the universal steepest-descent direction is the contrarian
policy $-(c-\bar c)$, while ascent yields momentum. For linear policies $\hat\pi(c)=Ac+b$, the gradient is the cost covariance matrix $\Sigma_c$, with a zero Hessian implying boundary-optimal solutions such as the minimum-variance portfolio. We extend to constrained optimization, sign-gradient duality between regret minimization and alpha maximization, finite-sample
convergence bounds paralleling Thompson Sampling, and gradient-descent algorithms requiring only input observations.
\end{abstract}


\section{Introduction and Motivation}

\cite{Aldridge2026} establishes that for linear optimization systems, the expected regret of any
policy $\hat\pi$ equals the covariance between uncertain costs and decisions:
\begin{equation}
\mathbb{E}[R(c)] = \mathrm{Cov}\big(c,\hat\pi(c)\big).
\label{eq:regret-identity}
\end{equation}
This identity transforms performance measurement into a computable scalar. The natural next
question is: in which direction should the policy change to reduce regret? Answering it requires
differentiating the covariance functional with respect to the policy.

This paper derives the complete derivative theory for $\mathrm{Cov}(c,\hat\pi(c))$ as a functional of the policy mapping $\hat\pi:\mathcal{C}\to\mathcal{Z}$: Sections~\ref{sec:derivative}--\ref{sec:sensitivity} derive its first- and second-order structure, constrained optima, and duality properties; Sections~\ref{sec:algorithms} and~\ref{sec:finite_sample} give gradient-based algorithms and finite-sample guarantees; and Section~\ref{sec:financial} presents financial applications.

\section{Related Work}
\label{sec:related}

The covariance identity $\mathbb{E}[R(c)] = \mathrm{Cov}(c, \hat\pi(c))$ of \cite{Aldridge2026} is a static optimization result, but the gradient algorithms and convergence guarantees we derive from it
(Sections~\ref{sec:algorithms}--\ref{sec:finite_sample}) connect to several strands of the no-regret learning literature.

\textbf{Online convex optimization.} Projected online gradient descent achieves $O(\sqrt{T})$ regret against an oblivious adversary \citep{zinkevich2003}, with $O(\log T)$ rates for strongly convex or exp-concave losses \citep{HazanAgarwalKale2007,ShalevShwartz2007}; see \cite{hazan2016} for a survey. Our covariance functional $C[A] = \mathrm{tr}(A\Sigma_c)$ is linear in  $A$, so the projected gradient step $A_{k+1} = \Pi_{\mathcal{Z}}(A_k - \eta\hat\Sigma_c)$ attains a linear convergence rate (Theorem~\ref{thm:convergence}) rather than the generic $O(1/\sqrt{T})$ OCO
bound because the functional has zero curvature (Theorem~\ref{thm:hessian}) and the gradient is estimable from cost observations alone, without access to decision outputs.

\textbf{Regret matching and CFR.} Counterfactual regret minimization \citep{ZinkevichEtAl2007} builds on the Regret matching dynamics of \cite{HartMasColell2000}, which converge to correlated equilibria by
weighting actions in proportion to cumulative regret. Our sign-gradient duality (Theorem~\ref{thm:sign_duality}) is structurally analogous: the regret-minimizing direction $-\Sigma_c$ and the alpha-maximizing direction $+\Sigma_c$ are opposite projections of the same gradient. Unlike RM/CFR,
however, our policy space is a function class evaluated under a fixed cost distribution, so our guarantee is a static convergence result rather than a per-round regret bound. 

\textbf{Thompson Sampling.} \cite{AgrawalGoyal2013} establish matching problem-dependent and problem-independent finite-sample regret bounds for Thompson Sampling. Section~\ref{sec:finite_sample} derives direct analogs for our setting, with the cost covariance gaps $\Delta_{ij} = [\Sigma_c]_{ij}$ playing the role of the arm gaps $\Delta_i$.

\textbf{Online portfolio selection.} Universal portfolios \citep{cover1991} and related methods \citep{HelmboldEtAl1998} minimize worst-case regret against an adversarial return sequence on the probability simplex. Our framework instead minimizes expected regret under a fixed cost distribution over a general feasible set (long-short, leverage, and factor constraints). 
 
\section{The Derivative of the Covariance Functional}
\label{sec:derivative}

\subsection{The functional and its domain}

Let $c \in \calC \subseteq \R^{d}$ be a random cost vector with
mean $\bar{c} = \E[c]$ and covariance matrix
$\Sig_c = \Cov(c,c) \in \R^{d \times d}$. A \emph{policy}
$\phat \in \calP$ is a measurable map from $\calC$ to the feasible
set $\calZ \subseteq \R^{n}$. The \emph{covariance functional} is
\begin{equation}\label{eq:cov_fun}
  \CovFun[\phat]
  \;=\; \Cov\!\bigl(c, \phat(c)\bigr)
  \;=\; \E\!\bigl[(c - \bar{c})^{\top}(\phat(c) - \E[\phat(c)])\bigr]
  \;\in\; \R.
\end{equation}
By \cite{Aldridge2026}, $\CovFun[\phat] = \E[\Regret]$ whenever
$\E[\phat(c)] = \pstar$. We study $\CovFun$ as an object to be
optimized over $\calP$.

\subsection{Gâteaux Derivative and the Universal Descent Direction}

The Gâteaux derivative of $C$ at $\hat\pi$ in the direction of $\phi: \mathcal{C} \to \mathbb{R}^n$ is
\begin{equation}
D C[\hat\pi](\phi) = \lim_{\epsilon \to 0} \frac{C[\hat\pi + \epsilon\phi] - C[\hat\pi]}{\epsilon}.
\end{equation}

\begin{theorem}[Gâteaux derivative of covariance regret]
\label{thm:gateaux}
Let $\hat\pi \in \mathcal{P}$ be any policy with finite second moments. For any direction
$\phi: \mathcal{C} \to \mathbb{R}^n$,
\begin{equation}
D C[\hat\pi](\phi) = \mathrm{Cov}\big(c, \phi(c)\big).
\end{equation}
\end{theorem}
\begin{proof}
Expanding $C[\hat\pi + \epsilon\phi] = \mathbb{E}[(c-\bar c)^\top(\hat\pi(c) + \epsilon\phi(c) -
\mathbb{E}[\hat\pi(c)] - \epsilon\mathbb{E}[\phi(c)])] = C[\hat\pi] + \epsilon\,\mathrm{Cov}(c,\phi(c))$, since $\mathbb{E}[c-\bar c]=0$ kills the cross term. Subtract $C[\hat\pi]$, divide by $\epsilon$, and
take the limit.
\end{proof}

Theorem~\ref{thm:gateaux} shows the derivative of the covariance functional is itself a covariance functional: the regret identity is preserved under differentiation, and the result is independent of the current policy $\hat\pi$. The functional gradient (Riesz representer) is therefore
$g(c) = c - \bar c$, giving:

\begin{corollary}[Universal steepest-descent direction]
\label{cor:universal}
The steepest-descent direction for $C$ in $L^2(P)$ is always $\phi^*(c) = -(c-\bar c)$, regardless of
$\hat\pi$; the steepest-ascent direction is $\phi^*(c) = +(c-\bar c)$.
\end{corollary}
\begin{proof}
By Cauchy--Schwarz, $DC[\hat\pi](\phi) = \mathrm{Cov}(c,\phi) \geq -\|\phi\|_{L^2}\|c-\bar c\|_{L^2}$,
with equality when $\phi(c) \propto -(c-\bar c)$.
\end{proof}

Economically, $\phi^*(c) = -(c-\bar c)$ is a contrarian/mean-reversion strategy (reduce the decision when costs are high), while $+(c-\bar c)$ is a momentum strategy — a first-principles derivation of why contrarian policies minimize regret and momentum policies maximize alpha.

\subsection{The Linear Case: Matrix Gradient}

For the finite-dimensional linear policy class $\hat\pi_\theta(c) = Ac+b$, $\theta=(A,b) \in
\mathbb{R}^{n\times d}\times\mathbb{R}^n$, Theorem~\ref{thm:gateaux} specializes directly. Since
$\mathrm{Cov}(c,b)=0$,
\begin{equation}
C[A,b] = \mathrm{Cov}(c, Ac+b) = \mathrm{tr}(A\Sigma_c).
\end{equation}

\begin{corollary}[Matrix gradient of covariance regret]
\label{cor:matrix-gradient}
For linear policies, $\nabla_A C = \Sigma_c$, $\nabla_b C = 0_n$, and the Hessian with respect to
$\mathrm{vec}(A)$ is $H_A C = \Sigma_c \otimes I_n$.
\end{corollary}
\begin{proof}
Immediate from Theorem~\ref{thm:gateaux} with $\phi(c) = \Delta A\, c$: $DC[\hat\pi](\Delta A\, c) =
\mathrm{Cov}(c, \Delta A\, c) = \mathrm{tr}(\Delta A\, \Sigma_c)$, identifying $\nabla_A C = \Sigma_c$.
The Hessian follows by vectorizing.
\end{proof}

Since $\Sigma_c \succeq 0$, gradient descent always pulls $A$ toward zero — the minimum-variance
policy. At the per-entry level,
\begin{equation}
\frac{\partial C}{\partial A_{ij}} = [\Sigma_c]_{ij} = \mathrm{Cov}(c_i,c_j),
\end{equation}
so regret is most sensitive along asset pairs with the largest cost covariance, and hedging those
pairs (driving the corresponding $A_{ij}$ toward zero) yields the greatest marginal regret reduction.

\section{Second-Order Structure: Convexity, Concavity, and Saddle Points}
\label{sec:second}

\begin{theorem}[Hessian of covariance regret]
  \label{thm:hessian}
  For linear policies $\phat(c) = Ac + b$, the covariance functional
  is \textbf{linear} (and hence simultaneously convex and concave) in
  $A$:
  \begin{equation}\label{eq:linear_in_A}
    \CovFun[\lambda A_1 + (1-\lambda)A_2, b]
    \;=\; \lambda\CovFun[A_1,b] + (1-\lambda)\CovFun[A_2,b]
    \quad\forall\,\lambda \in [0,1].
  \end{equation}
  Therefore $\CovFun$ has no interior critical points in the
  unconstrained linear-policy class: regret is minimized (or
  maximized) on the boundary of the feasible set.
\end{theorem}

\begin{proof}
  $\CovFun[A,b] = \tr(A\Sig_c)$ is linear in $A$, so it is both convex and
  concave. All second derivatives are zero.
\end{proof}

\begin{theorem}[Nonlinear policies: second-order conditions]
  \label{thm:nonlinear_second}
  For a general (possibly nonlinear) policy $\phat$ parametrized by
  $\theta \in \R^p$, the Hessian of $\CovFun$ with respect to $\theta$
  evaluated at $\theta_0$ is:
  \begin{equation}\label{eq:nonlin_hess}
    \Hess_\theta\CovFun\big|_{\theta_0}
    \;=\; \E\!\left[
          (c - \bar{c})^{\top}
          \frac{\partial^2\phat}{\partial\theta\partial\theta^{\top}}
          \right].
  \end{equation}
  A critical point $\theta^*$ (where $\grad_\theta\CovFun = 0$) is:
  \begin{itemize}
    \item a \textbf{local minimum} (regret-minimizing policy) if
          $\Hess_\theta\CovFun\big|_{\theta^*} \succcurlyeq 0$;
    \item a \textbf{local maximum} (alpha-maximizing policy) if
          $\Hess_\theta\CovFun\big|_{\theta^*} \preccurlyeq 0$;
    \item a \textbf{saddle point} if the Hessian is indefinite.
  \end{itemize}
\end{theorem}

\begin{proof}
  Differentiate $\grad_\theta\CovFun = \E[(c-\bar c)^{\top}
  \partial\phat/\partial\theta]$ again with respect to $\theta$,
  exchanging differentiation and expectation by dominated convergence
  (assuming the policy is twice differentiable with bounded second
  derivatives). The result follows.
\end{proof}

\section{Constrained Optimization: The Regret Lagrangian}
\label{sec:constrained}

In portfolio optimization, the policy must satisfy feasibility
constraints. We derive the constrained gradient and the conditions
for an interior optimum.

\subsection{General constrained problem}

The \emph{regret minimization problem} is:
\begin{equation}\label{eq:reg_min}
  \min_{\phat \in \calP}\;\CovFun[\phat]
  \quad\text{subject to}\quad
  g_k(\phat) \leq 0,\; k = 1,\ldots,K.
\end{equation}
The Lagrangian is:
\begin{equation}\label{eq:lagrangian}
  \calL[\phat, \lambda]
  \;=\; \Cov(c, \phat(c))
       + \sum_{k=1}^{K}\lambda_k g_k(\phat).
\end{equation}
KKT conditions: $D\calL[\phat^*, \lambda^*](\phi) = 0$ for all
feasible $\phi$, complementary slackness $\lambda_k^* g_k(\phat^*) = 0$,
and dual feasibility $\lambda_k^* \geq 0$.
\subsection{Budget and Return Constraints}
\label{sec:constraints}

We solve the general constrained problem for two practically important feasible sets: the minimum-variance portfolio (budget constraint alone) and
the return-targeted frontier (budget plus a return floor). Both follow from the same KKT system applied to the linear policy $w = Ac+b$.

For the budget constraint $\mathbf{1}^\top w = 1$ with non-negativity $w \geq 0$, introduce a multiplier $\mu$ for the budget and $\nu \geq 0$ for non-negativity:
\begin{equation}
L[A,b,\mu,\nu] = \mathrm{tr}(A\Sigma_c) + \mu\big(\mathbf{1}^\top(A\bar c+b)-1\big) -
\nu^\top(A\bar c+b).
\end{equation}
Adding a return target $\bar r^\top(A\bar c+b)\geq\mu_0$ introduces a second multiplier $\kappa$, giving rise to the general stationarity conditions
\begin{equation}
\nabla_A L = \Sigma_c + (\mu\mathbf{1}-\nu)\bar c^\top - \kappa\bar r\bar c^\top = 0, \qquad
\nabla_b L = \mu\mathbf{1}-\nu = 0.
\label{eq:kkt-general}
\end{equation}

\begin{corollary}[Minimum-variance portfolio]
\label{cor:mvp}
With $\kappa=0$ (no return constraint), \eqref{eq:kkt-general} reduces to $\Sigma_c=0$ unless the minimum is attained at the boundary $A^*=0$. For $\Sigma_c\succ0$, the budget-constrained regret minimum over linear policies is
\begin{equation}
A^*=0, \qquad b^* = \Sigma_c^{-1}\mathbf{1}\big/(\mathbf{1}^\top\Sigma_c^{-1}\mathbf{1}),
\end{equation}
the minimum-variance portfolio.
\end{corollary}
\begin{proof}[Proof sketch]
$C[A,b]=\mathrm{tr}(A\Sigma_c)$ is linear and increasing along any direction where $\Sigma_c\succ0$, so the constrained minimum over the budget simplex sits at $A=0$; the resulting problem in $b$ subject to $\mathbf{1}^\top b=1$ is the standard MVP quadratic program.
\end{proof}

\begin{corollary}[Regret-return efficient frontier]
\label{cor:frontier}
With $\kappa>0$, \eqref{eq:kkt-general} gives $A^* = \kappa\,\Sigma_c^{-1}\bar r\bar c^\top$, and the set of policies minimizing regret subject to a return target $\mu_0$ traces out a frontier
\begin{equation}
C^*(\mu_0) = \kappa^*\,\bar c^\top\Sigma_c^{-1}\bar r\cdot\mathrm{tr}(\Sigma_c),
\end{equation}
parameterized by $\kappa^*\geq0$ chosen so that the return constraint binds. This is the covariance-based analog of the Markowitz mean-variance frontier.
\end{corollary}
\begin{proof}[Proof sketch]
Set $\nabla_A L = \Sigma_c - \kappa\bar r\bar c^\top = 0$ and solve for $A^*$; substitute into $C[A^*,b]$ and choose $\kappa$ so that the return constraint binds.
\end{proof}

\section{Sign-Gradient Duality: Minimizing Loss vs.\ Maximizing Alpha}
\label{sec:sign}

\begin{theorem}[Sign-gradient duality]
  \label{thm:sign_duality}
  For a cost-minimization problem, the regret-minimizing gradient
  direction is $-\Sig_c$ (reduce cost sensitivity). For the
  equivalent return-maximization problem with $r = -c$:
  \begin{equation}\label{eq:alpha_grad}
    \grad_A\Cov(r, \phat(r))
    \;=\; \Sig_r \;=\; \Sig_c.
  \end{equation}
  The \textbf{alpha-maximizing gradient direction is $+\Sig_c$}
  (increase return sensitivity). The identical gradient magnitude
  with opposite signs reflects the sign-gradient duality:
  \begin{equation}\label{eq:duality}
    \underbrace{\text{minimize regret}}_{\text{subtract }\eta\Sig_c}
    \;\longleftrightarrow\;
    \underbrace{\text{maximize alpha}}_{\text{add }\eta\Sig_c}.
  \end{equation}
\end{theorem}

\begin{proof}
  For $r = -c$, $\Cov(r, Ar + b) = \tr(A\Sig_r) = \tr(A\Sig_c)$
  (since $\Sig_r = \Sig_{-c} = \Sig_c$). Maximizing this is
  equivalent to following the gradient $+\Sig_c$, while minimizing
  regret follows $-\Sig_c$.
\end{proof}

\section{Sensitivity Analysis: Cross-Derivatives}
\label{sec:sensitivity}
\begin{remark}[Regret sensitivity to volatility and mean]
\label{rem:sensitivity}
Since $C[A,b]=\mathrm{tr}(A\Sigma_c)$ depends on $\Sigma_c$ but not on $\bar c$, regret is insensitive to the cost mean ($\partial C/\partial\bar c=0$) and linearly sensitive to the cost covariance, with the sensitivity matrix $A$ itself, $\partial C/\partial\Sigma_c=A$; in the isotropic
case $\Sigma_c=\sigma_c^2 I$, this gives volatility elasticity $\partial C/\partial\sigma_c^2=\mathrm{tr}(A)$.
\end{remark}

\section{Gradient-Based Algorithms for Policy Optimization}
\label{sec:algorithms}

\subsection{Projected gradient descent for regret minimization}
\label{alg:pgd}

Given a step size $\eta>0$, an initial policy $A_0$, and a feasible set $\mathcal{Z}$, covariance gradient descent proceeds by observing $N$ cost realizations $c_1,\dots,c_N$, estimating $\hat\Sigma_c = \frac{1}{N-1}\sum_i (c_i-\bar c)(c_i-\bar c)^\top$, and updating
$A_{k+1}=\Pi_{\mathcal{Z}}(A_k-\eta\hat\Sigma_c)$ until $\|A_{k+1}-A_k\|<\varepsilon$. The gradient step $-\eta\hat\Sigma_c$ is computable from input data alone — no output observations are needed for
the gradient, only for estimating $\hat\Sigma_c$.
\subsection{Variants: Natural and Confidence-Weighted Gradients}

Two variants extend Algorithm~\ref{alg:pgd} along complementary axes: the natural policy gradient pre-multiplies the update by the inverse Fisher information metric, reducing to $A_{k+1}=A_k-\eta\Omega$ for linear Gaussian policies and yielding scale invariance across cost distributions, while the confidence-weighted gradient of \cite{ChenDidisheimSomoza2026} re-weights observations by an entropy-based confidence score to concentrate updates on the most reliable cost signals.

\subsection{Convergence rate}

\begin{theorem}[Convergence of covariance gradient descent]
  \label{thm:convergence}
  For the unconstrained linear-policy problem with step size
  $\eta \leq 1/\norm{\Sig_c}_2$, the projected gradient descent
  of Algorithm~\ref{alg:pgd} satisfies:
  \begin{equation}\label{eq:convergence}
    \CovFun[A_k] - \CovFun[A^*]
    \;\leq\; (1 - \eta\lambda_{\min}(\Sig_c))^k
             (\CovFun[A_0] - \CovFun[A^*]).
  \end{equation}
  The convergence rate is determined by the condition number
  $\kappa(\Sig_c) = \lambda_{\max}(\Sig_c)/\lambda_{\min}(\Sig_c)$
  of the cost covariance matrix. Ill-conditioned portfolios
  (near-collinear assets) converge slowly; diversified portfolios
  converge rapidly.
\end{theorem}

\begin{proof}
  Since $\CovFun[A,b] = \tr(A\Sig_c)$ is linear in $A$ with
  gradient $\Sig_c$, the projected gradient step contracts the
  distance to $A^*$ by a factor of $(1-\eta\lambda_{\min}(\Sig_c))$
 at each iteration, yielding \eqref{eq:convergence} by induction.
\end{proof}

\section{Finite-Sample Regret Bounds}
\label{sec:finite_sample}

The convergence rate in Theorem~\ref{thm:convergence} is asymptotic. Practitioners need to know how many gradient steps and how many cost observations are required to bring regret below a target
tolerance. We answer both questions, in direct analogy with the finite-time analysis of Thompson Sampling by \cite{AgrawalGoyal2013}, who bound cumulative bandit regret in terms of arm-gap parameters
$\Delta_i$. 
\subsection{Setup}

Recall $C[A,b] = \mathrm{tr}(A\Sigma_c)$ and $\partial C/\partial A_{ij} = [\Sigma_c]_{ij}$
(Corollary~\ref{cor:matrix-gradient}). The optimum is $A^*=0$ (Theorem~\ref{thm:hessian}), giving
$C[A^*,b^*]=0$.

\begin{definition}[Asset-pair gap]
For $i,j \in \{1,\dots,d\}$, define $\Delta_{ij} := [\Sigma_c]_{ij}$, with minimum and maximum eigenvalue gaps $\Delta_{\min} := \lambda_{\min}(\Sigma_c) > 0$, $\Delta_{\max} :=
\lambda_{\max}(\Sigma_c)$, and condition number $\kappa := \Delta_{\max}/\Delta_{\min}$.
\end{definition}

The gaps $\Delta_{ij}$ play the role of the bandit arm gaps $\Delta_i$: they measure how much regret is generated per unit of allocation weight $A_{ij}$, with regret being most sensitive along the directions of the largest $|\Delta_{ij}|$.

\subsection{Problem-Dependent Bound}

We assume cost vectors are sub-Gaussian: $\mathbb{E}[e^{s^\top(c-\bar c)}] \leq e^{\sigma^2\|s\|^2/2}$ for some $\sigma>0$ and all $s \in \mathbb{R}^d$.

\begin{theorem}[Problem-dependent finite-sample bound]
\label{thm:finite-sample}
Let $\Sigma_c \succ 0$ with minimum eigenvalue $\Delta_{\min}$ and condition number $\kappa$. Run Algorithm~\ref{alg:pgd} with step size $\eta = 1/\|\Sigma_c\|$ and the true gradient $\Sigma_c$. After $K$ steps, starting from any $A_0$, 
\begin{equation}
C[A_K] - C[A^*] \leq \left(1 - \tfrac{1}{\kappa}\right)^K C[A_0].
\end{equation}
To achieve excess regret of at most $\varepsilon$, it suffices to take
\begin{equation}
K(\varepsilon) = \kappa \ln\!\big(C[A_0]/\varepsilon\big)
\end{equation}
gradient steps.
\end{theorem}
\begin{proof}[Proof sketch]
$C[A]=\mathrm{tr}(A\Sigma_c)$ is linear with a gradient $\Sigma_c$, so $\|A_{k+1}\| \leq
\|A_k\|(1-\eta\lambda_{\min}(\Sigma_c))$; setting $\eta=1/\Delta_{\max}$ gives the contraction factor $1-1/\kappa$. Since $C$ inherits the same geometric contraction, solving for $K$ gives the bound. \end{proof}

\subsection{Sample Complexity}

In practice, $\Sigma_c$ must be estimated from $N$ i.i.d. cost observations via $\hat\Sigma_c = \frac{1}{N-1}\sum_i (c_i-\bar c)(c_i-\bar c)^\top$.

\begin{theorem}[Sample complexity for gradient estimation]
\label{thm:sample-complexity}
For sub-Gaussian costs with parameter $\sigma^2$, any $\delta \in (0,1)$, and target gradient error $\xi > 0$, a sample of size
\begin{equation}
N \geq C \,\frac{\sigma^4 d \ln(d/\delta)}{\xi^2}
\end{equation}
(universal constant $C>0$) guarantees $\|\hat\Sigma_c - \Sigma_c\| \leq \xi$ with a probability of at least $1-\delta$. Running Algorithm~\ref{alg:pgd} with $\hat\Sigma_c$ and step size
$\eta = 1/(\|\Sigma_c\|+\xi)$, the stochastic iterates satisfy
\begin{equation}
\mathbb{E}\big[C[A_K]\big] - C[A^*] \leq \left(1 - \tfrac{\Delta_{\min}}{\Delta_{\max}+\xi}\right)^K
C[A_0] \;+\; \frac{\xi\|A_0\|_F}{\Delta_{\min}},
\end{equation}
where the first term decays geometrically and the second is a statistical bias vanishing as
$N \to \infty$.
\end{theorem}
\begin{proof}[Proof sketch]
The spectral-norm bound follows from the matrix Bernstein inequality applied to the rank-1 summands
$(c_i-\bar c)(c_i-\bar c)^\top$ \citep{tropp2015}. The convergence bound follows by bounding the
perturbation $\mathrm{tr}((\Sigma_c-\hat\Sigma_c)\Sigma_c)$ in the one-step update and applying a union bound over $K$ steps.
\end{proof}


\section{Numerical Experiments -- Shortest Path}
Here, we replicate the $5\times5$ grid shortest path experiments of \cite{ElmachtoubGrigas2022} and \cite{HuangGupta2024}.  Every edge in the grid is assigned a cost that is based on predictors $X_{0,i}$, coefficients $B_{ij}$ and error $\epsilon_i$:
    \begin{equation}\label{eq:costs_5x5}
    c_{ij}(X) = \frac{1}{\sqrt{p}}\left[(B_{ij} X_{0,i})+3)^{deg}+1\right](1+\epsilon_i)^j    
    \end{equation}

All experiments were conducted on a basic Windows platform with an Intel(R) N100 (800 MHz) processor and 16GB RAM, using standard Python. No accelerators were used.

\subsection{Why $\deg-1$ Is Loss-Minimizing: A Bias--Regret Decomposition}
\label{sec:deg-minus-one}

We derive, using the paper's Corollary~3.3 and Theorem~4.1, why the 
underfit model $m = \deg-1$ achieves lower total loss than the correctly specified model $m = \deg$, even though it is a strictly worse point predictor of costs.

\paragraph{Loss decomposition.}
For a fitted cost model of polynomial order $m$, the decision maker computes
$\hat c^{(m)}(X) = \mathbb{E}[c \mid X; m]$ and acts on $z^*(\hat c^{(m)})$. Realized loss against the oracle decomposes as
\begin{equation}
\mathrm{Loss}(m) \;=\; \underbrace{\big(\mathbb{E}[c\mid X;m] - \mathbb{E}[c\mid X]\big)^2}_{\mathrm{Bias}(m)^2}
\;+\; \underbrace{\operatorname{Cov}\!\big(c,\, z^*(\hat c^{(m)}(c))\big)}_{C[A_m]},
\label{eq:decomp}
\end{equation}
where the second term is exactly the covariance regret functional of Section~3.

\paragraph{Local linearization.}
Near an operating point $\bar X$, the cost model
$c_{ij}(X) \propto (B_{ij}X_{0,i}+3)^m$ has a local slope.
\begin{equation}
\delta_m \;=\; \frac{\partial}{\partial X}\big(BX+3\big)^m\Big|_{\bar X}
\;=\; m\,B\,(B\bar X + 3)^{m-1}.
\end{equation}
If $z^*(\cdot)$ is locally linear in the plugged-in cost, the induced decision rule is
approximately $z^*(\hat c^{(m)}(X)) \approx A_m X + b_m$, with $A_m$ scaling with $\delta_m$. Since $B\bar X + 3 > 1$ in the experimental regime, $\delta_m$ is \emph{increasing} in $m$: higher-degree fits induce a decision rule with strictly larger effective sensitivity $A_m$.

\paragraph{Regret is linear and increasing in $A$.}
By Corollary~3.3, $C[A_m] = \operatorname{tr}(A_m \Sigma_c)$ with
$\nabla_A C = \Sigma_c \succeq 0$. Because $C$ is \emph{linear} in $A_m$ (Theorem~4.1), increasing decision sensitivity always adds regret, regardless of whether that sensitivity improves the fit to the true cost mean.

\subsection{Empirical Results}
 Figures \ref{fig:avgPerformance-low-deg1} and \ref{fig:avgPerformance-low-deg2} show the average loss for different models calculated on 10,000 simulated samples when the actual polynomial degree $deg$ is not known, as in \cite{ElmachtoubGrigas2022} and \cite{HuangGupta2024}. The models are re-estimated continuously as the number of available observations increases. The optimal decision based on the expected cost with $deg-1$ tends to outperform other methods (realized loss, $z^*(c_{E[deg]-1})$). SPO (SPO+ (Ridge) in the graphs) often tracks the expected loss calculated as the covariance of simulated or historical costs and the optimal decisions based on those costs (Ex-ante E [Loss]: $covar(z^*(E[c]), E[c])$). 
 
 \begin{figure}
    \centering
    {\includegraphics[width=0.95\linewidth]{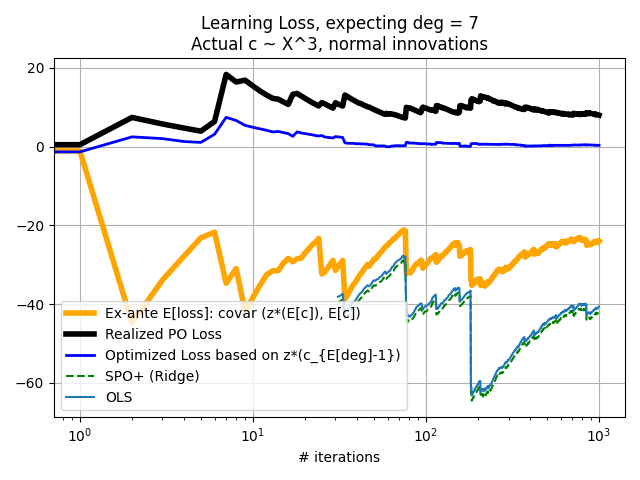}}
    \caption{Average performance of different loss models over 100,000 samples in the shortest-path models with randomly selected $X$ and $B$.}
    \label{fig:avgPerformance-low-deg1}
\end{figure}

 \begin{figure}
    \centering
    { \includegraphics[width=0.95\linewidth]{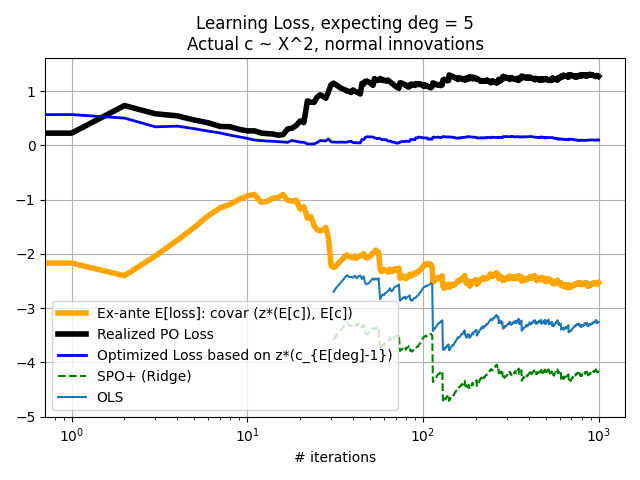}}
    \caption{Average performance of different loss models over 100,000 samples in the shortest-path models with randomly selected $X$ and $B$.}
    \label{fig:avgPerformance-low-deg2}
\end{figure}

\section{Financial Application: Optimal tilt away from the minimum-variance portfolio}
\label{sec:financial}

The MVP has $A = 0$ and zero regret. An asset manager who wishes
to take active bets (tilt toward expected alpha) does so by adding
$\Delta A \neq 0$. The regret cost of a tilt $\Delta A$ is:
\begin{equation}\label{eq:tilt_cost}
  \Delta\CovFun \;=\; \tr(\Delta A\,\Sig_c),
\end{equation}
which is the trace of the product of the tilt matrix and the cost
covariance. The \emph{regret-optimal tilt} toward a target
allocation $w_{\mathrm{target}}$ is:
\begin{equation}\label{eq:opt_tilt}
  \Delta A^* \;=\; \arg\min_{\Delta A:\,\Delta A\bar c = w_{\mathrm{target}} - w_{\mathrm{MVP}}}
  \tr(\Delta A\,\Sig_c),
\end{equation}
whose solution is $\Delta A^* = (w_{\mathrm{target}} - w_{\mathrm{MVP}})\bar{c}^{\top}
(\bar{c}\bar{c}^{\top})^{-1}$, a rank-1 update.

\section{Conclusion}

The Gâteaux derivative of the covariance-based regret functional is itself a covariance, preserving the self-referential structure of the regret identity under differentiation. This yields a universal, policy-independent descent direction $\phi^*(c)=-(c-\bar c)$ — the contrarian/mean-reversion strategy, with the momentum strategy $+(c-\bar c)$ as its ascent
counterpart — and, for linear policies $\hat\pi=Ac+b$, a matrix gradient equal to the cost covariance $\Sigma_c$, with the minimum-variance portfolio ($A=0$) as the unique unconstrained minimizer and the maximum-leverage momentum policy as its maximizer. Minimizing loss and
maximizing alpha are gradient descent and ascent on the same functional, with the unified update $A \leftarrow A \mp \eta\Sigma_c$. Because the gradient $\hat\Sigma_c$ is estimable from input data alone, with no output observations required, gradient descent on regret is a zero-instrumentation
algorithm, consistent with the black-box evaluation framework of \cite{Aldridge2026}.

\bibliographystyle{alpha}
\bibliography{OptimizingRegret}


\end{document}